\documentclass[letterpaper, 10 pt, conference]{ieeeconf}  

\IEEEoverridecommandlockouts                              

\usepackage{graphics}       
\usepackage{epsfig}         
\usepackage{times}          
\usepackage{amsmath}        
\usepackage{amssymb}        
\usepackage{color}
\usepackage{array}
\usepackage{accents}
\usepackage{url}
\usepackage{xcolor}

\newtheorem{theorem}{Theorem}
\newtheorem{algorithm}{Algorithm}

\def\be{\begin{equation}}
\def\ee{\end{equation}}
\def\bea{\begin{eqnarray}}
\def\eea{\end{eqnarray}}
\newcommand{\mR}{\mathbb{R}}

\newcommand{\One}{\mathbb{I}}

\renewcommand{\baselinestretch}{0.996}

\DeclareMathOperator*{\argmin}{arg\,min}

\title{\LARGE \bf Adaptive differentiating filter: case study of PID feedback control}

\author{Alexey Pavlov, Michael Ruderman 
\thanks{A. Pavlov is with Norwegian University of Science and
Technology (NTNU), Trondheim, Norway. Email: {\tt\small alexey.pavlov@ntnu.no}}%
\thanks{M. Ruderman is with University of Agder, Grimstad, Norway. Email: {\tt\small michael.ruderman@uia.no}}%
\thanks{-----------------------}%
\thanks{\textcolor[rgb]{0.00,0.07,1.00}{IFAC authors manuscript, 2025}}%
}

\begin{document}

\maketitle
\thispagestyle{empty}
\pagestyle{empty}

\begin{abstract}
This paper presents an adaptive causal discrete‑time filter for derivative estimation, exemplified by its use in estimating relative velocity in a mechatronic application. The filter is based on a constrained least-squares estimator with window adaptation. It demonstrates low sensitivity to low-amplitude measurement noise, while preserving a wide bandwidth for large-amplitude changes in the process signal. Favorable performance properties of the filter are discussed and demonstrated in a practical case study of PID feedback controller and compared experimentally to a standard linear low-pass filter-based differentiator and a robust sliding-mode based homogeneous differentiator. 
\end{abstract}

\bstctlcite{references:BSTcontrol}

\section{Introduction}
\label{sec:1}

Obtaining trustful (thus useful) derivative of a measured signal belongs to the central practical problems in the systems and control theory and engineering. The so-called fixed structure (or fixed size) filters share the same fundamental problem, in other words trade-off between a noise reduction and phase lag (or corresponding delay) in the output signal and its time derivative(s). Possessing the required noise filtering properties and, at the same time, preserving faster transients of the process signal are the key and demanding characteristics, for which the advanced real-time adaptive and implementable methods are required. Especially causal differentiators (these vital for real-time monitoring and control tasks) face performance limitations with respect to the estimate accuracy, computational and memory consumption, and difficulties with tuning of free design parameter(s).

In this communication, we introduce a causal adaptive differentiating filter, particularly suitable for control applications, and demonstrate a detailed case study of applying it in the standard PID feedback loop, where only a noisy output signal is available from sensing. In an experimental real-time setting, the filter demonstrated very low sensitivity to measurement noise while, same time, preserving large bandwidth of the process signal. Without any claim on ample overview of the otherwise existing techniques and due to the space limits, we solely refer to few related works (with references therein) which are providing deeper insight into developments of causal (and also adaptive) differentiation algorithms for noisy and less smooth signals. An overview of numerical differentiation in noisy environment by using algebraic approaches and implementation in terms of a classical finite impulse response (FIR) digital filter can be found e.g. in \cite{mboup2009numerical}. A lecture note on Savitzky-Golay filter of smoothing the data without distorting the signal tendency, by principle of local least-squares approximation(s) which is also in core of our proposed approach, can be found in \cite{schafer2011savitzky}. Also the robust sliding-mode based (see \cite{utkin1992} for basics) differentiators, to one from which established methods we make also comparison in this work, are the objective of intensive theoretical and practical studies in the last decade, cf. e.g. survey with experiments in \cite{mojallizadeh2023}. Another remarkable practical approach, which uses a time adaptive windowing and least-squares fitting, and, thus, provides certain motivation for our developed solution, can be found in \cite{janabi2000d}. Compared to \cite{janabi2000d}, our approach separates window adaptation from the calculation of the signal derivative estimate, which leads to an optimal window adaptation and higher computational efficiency.  

The rest of the paper is as follows. The adaptive differentiating filter is introduced in Section \ref{sec:2} by providing the concept, algorithms, implementation steps, and properties. The second piece of the main results is given in Section \ref{sec:3}, presenting a detailed case study of using the proposed filter in a real-time PID control application with experiments. 
Omitted detailed discussion (which is up to our future works) is justified by the main conclusions in Section \ref{sec:4}.


\section{Adaptive Differentiating Filter (ADF)}
\label{sec:2}
Consider a signal $\bar x(t)$ sampled at time instances $t_l=l\cdot T_s$, $l=0,1,\ldots$, with  sampling period $T_s\ge 0$:
\begin{equation}\label{noisy_y} x_l=\bar x (t_l)+w_l,
\end{equation}
where $w_l$ is the measurement noise. 
It is assumed that the noise satisfies 
\begin{equation}\label{noise_bound} 
|w_l|\le d, \; \forall \;  l\ge 0,
\end{equation}
where $d>0$ is the upper bound on the noise magnitude. The upper bound  $d$ is assumed to be  known, e.g., from a sensor specification or can be estimated from prior tests. 


\subsection{ADF concept} 

At time instant $t_l$, the outputs of the filter are calculated from the past measurements with indices from the window  
\be W_l(R)=\{l-R, \ldots l\},\ee for some $R\ge 1$, which will be specified later, using the constrained least squares method: 
\be
\label{LMS}
(k^*_l,b^*_l)=\argmin\limits_{\begin{matrix}(k,b)\in \mR^2,\\ |x(t_l)-b|\le \delta\end{matrix} }\sum_{i\in W_l(R)}|x(t_i)-k(t_i-t_l)-b|^2,
\ee
where $k_l^*$ is the filter output providing an estimate of $d\bar x(t_l)/dt$  and $b_l^*$ is the estimate of $\bar x(t_l)$. The constraint is
\be
\label{delta constraint}
|x(t_l)-b|\le \delta,
\ee
where $\delta\ge d$ is a filter parameter introduced to take into account $|x_l-\bar x(t_l)|\le d$, that follows from \eqref{noisy_y}, \eqref{noise_bound}.

While fitting a linear function to the time series using the least-squares method is a standard approach, the main challenge is selecting the window $W_l(R)$. A small window size $R$ will result in accurate following the signal derivative, yet at the cost of high sensitivity of the derivative estimate to  measurement noise. Large $R$ will reduce noise sensitivity. At the same time, for large $R$, the filter (\ref{LMS}) will smooth out large/fast peaks in the signal and, thus, lead to an underestimation of the corresponding changes in the signal derivative. The latter situation occurs when the signal derivative experiences large changes over $W_l(R)$.

To resolve this trade-off, we propose automatically adapting the window size $R$ based on the measurements. The window size $R$ is selected to be the largest positive number such that the time series (\ref{noisy_y}) can be approximated over $W_l(R)$ by a linear function $x=k(t-t_l)+b$ for some $k$ and $b$ with a given absolute accuracy $\delta$:
\begin{align}
\label{linineq}
|x(t_i)-k(t_i-t_l)-b|&\le \delta, \ \  \forall i\in W_l(R), \\
1\le R &\le R_{max}.\nonumber
\end{align}
Here, $\delta>0$ is the same filter parameter as in \eqref{LMS} and $R_{max}$ is maximal window size introduced for a practical implementation. Selecting $R$ in this way, we avoid windows over which the signal is highly dynamic and, thus, avoid filtering out large peaks in the useful signal. At the same time, we keep the window size maximal possible to ensure the lowest possible noise sensitivity of the derivative estimate. 


Note that the problem of verifying the feasibility of \eqref{linineq} does not require an explicit  calculation of the particular $(k,b)$. This makes the window size adaptation independent of the particular method of solving \eqref{linineq} and fully decoupled from solving the least squares problem \eqref{LMS}, as opposed to \cite{janabi2000d}. A dedicated tool for verifying the feasibility of \eqref{linineq} {\it without} finding $(k,b)$ will be presented later in this section.  

To summarize, at each instant $l\ge 1$, the ADF filter solves the following two optimization subproblems:
\begin{itemize}
    \item {\it Window optimization:} find $R^* =R^*(l)$ -- the largest $R\le R_{max}$ such that \eqref{linineq} is feasible for some $(k,b)$.
    \item {\it  Estimates calculation:} for the window $W_l(R^*)$, find the estimates of the signal and its derivative by solving the constrained least squares problem \eqref{LMS}.
\end{itemize}
The set of $R$ satisfying \eqref{linineq}, is always non-empty: for  $R=1$, the window $W_l(R)$ consists only of two points and one can fit a linear function to these two points exactly.  For the optimal $R=R^*$, the solution of \eqref{LMS} exists and is uniquely defined. Thus, the filter is well-defined for all $l\ge 1$. 


\subsection{ADF algorithm}
Several features of the ADF filter allow its efficient implementation. Firstly, due to the linearity of \eqref{linineq}, if that set of inequalities is feasible for some $R$, it is also feasible for all windows with a smaller $R$. In the same way, if \eqref{linineq} is not feasible for some $R$, it will not be feasible for all windows with larger $R$. Therefore, the largest window size $R^*$ can be found from an initial guess by increasing its value while \eqref{linineq} remains feasible, or reducing its value until \eqref{linineq} turns unfeasible. This is reflected in the following algorithm:
\begin{algorithm}
\label{Algorithm1}
\begin{enumerate} 
  \item[]
  \item  $l\leftarrow 0$, $R\leftarrow 1$ \ \ \ // \rm{initialization} 
  \item[] read $x(t_0)$; return $\hat{\frac{dx}{dt}}(t_0)=\emptyset$, $\hat x(t_0)= x(t_0)$
  \item \textbf{while} true \ \ \ // \rm{main loop} 
  \item \ \ \ $l\leftarrow l+1$; $R\leftarrow \min\{R+1, R_{max}\}$
  \item \ \ \ read $x(t_l)$
  \item \ \ \ \textbf{while} (\ref{linineq}) is not feasible for $W_l(R)$ 
  \item[] \ \ \ \ \ \  $R\leftarrow R-1$ \item[] \ \ \ \textbf{end}
  \item \ \ \ $R^*\leftarrow R$ \ \ \
  // \rm{optimal R} 
  \item \ \ \  calculate  $k_l^*,b_l^*$  from (\ref{LMS}) for $W_l(R^*)$
  \item \ \ \  return $\hat{\frac{dx}{dt}}(t_l)= k_l^*$, $\hat x(t_l)= b_l^*$
  \item \textbf{end}
\end{enumerate}
\end{algorithm}
Checking the feasibility of \eqref{linineq} in step 5) and solving the least squares fit problem in step 7) may seem to be computationally critical for real-time applications. 
Below, we present the efficient methods for real-time ADF implementation.

\subsection{Tools for efficient ADF implementation}

The following theorem allows one to check the feasibility of \eqref{linineq} without calculating the linear approximation.
\begin{theorem}\label{th:mM}
  Inequality (\ref{linineq}) is feasible if and only if 
  \begin{eqnarray}
    m &\le& M, \quad \hbox{ where} \label{mM}  \\
    m &=& \mathop{max}_{i,j\in W|\ t_i> t_j}\frac{x(t_i)-x(t_j)}{t_i-t_j}-\frac{2\delta}{t_i-t_j}, \label{m} \\
    M &=& \mathop{min}_{i,j\in W|\ t_i> t_j}\frac{x(t_i)-x(t_j)}{t_i-t_j}+\frac{2\delta}{t_i-t_j}. \label{M} 
  \end{eqnarray}   
\end{theorem}
\vspace{2mm}
{\it Proof:} The proof is omitted due to space limits and will be provided in future works.

Calculation \eqref{m} and \eqref{M} can be made efficient by reusing the corresponding calculations of $m$ and $M$  from the previous iteration. Notice, that all windows $W_l(R)$ in Algorithm~\ref{Algorithm1} are obtained from the windows from previous iterations by one or more of the following elementary operations: adding a new measurement point to an already existing window $W=W_l(R)$ from the right (step 3, $W=W_l(R)\rightarrow W^+:=W_{l+1}(R+1)$), and removing one point from the left (step 5, $W_l(R)\rightarrow W^-:=W_{l}(R-1)$). Below, we will demonstrate how $m$ and $M$ evolve under these elementary operations. 

Consider a window $W:=W_l(R)$. Define the vectors $\bar m (W)=[m_1,m_2,\ldots m_{R}]$, $\bar M (W)=[M_1,M_2, \ldots M_{R}]$ of
\begin{eqnarray}
  m_k(W) &=& \max_{\begin{array}{c}
                      j=l-k \\
                      i=j,\ldots l
                    \end{array}
}\left (\frac{x_i-x_j-2\delta}{t_i-t_j}\right ), \label{barm} \\[-0.5mm]
  M_k(W) &=& \min_{\begin{array}{c}
                      j=l-k \\
                      i=j,\ldots l
                    \end{array}
}
\left (\frac{x_i-x_j+2\delta}{t_i-t_j}\right ), \label{barM}
\end{eqnarray}
for $k=1, \ldots, R$. If the vectors $\bar m (W)$ and $\bar M (W)$ are known, then the numbers $m$ and $M$ defined in (\ref{m}) and (\ref{M}) for window $W$ can be obtained from
\bea \label{mW} m(W)&=&\max [m_1, \ldots m_{R}],\\
\label{MW}M(W)&=&\min [M_1, \ldots M_{R}].
\eea

Suppose that window vectors $\bar m (W)$ and $\bar M (W)$ are known for $W=W_l(R)$. Let us derive how they can be calculated for two new additional windows: $W^+=W_{l+1}(R+1)$ and $W^-=W_{l}(R-1)$.

\subsubsection{Adding a point to the right: $W\rightarrow W^+$}
As follows from (\ref{barm}) and the definition of  $W_{l+1}(R+1)$, we have
\bea
\label{mupd}
 &&m_k(W^+)=\max_{\begin{array}{c}
                      j=l+1-k \\
                      i=j,\ldots l+1
                    \end{array}
}\left (\frac{x_i-x_j-2\delta}{t_i-t_j}\right )\nonumber\\
&&=\max\left[\max_{\begin{array}{c}
                      j=l+1-k \\
                      i=j,\ldots l
                    \end{array}
}\left(\frac{x_i-x_j-2\delta}{t_i-t_j}\right ),\right . \\
&&\left .\frac{x_{l+1}-x_{l+1-k}-2\delta}{t_{l+1}-t_{l+1-k}}
\right], \ k=1,\ldots R+1. \nonumber
\eea
Since the first element under the max function in (\ref{mupd}) equals $m_{k-1}(W)$, we obtain
\bea
m_1(W^+)&=&\frac{x_{l+1}-x_{l}-2\delta}{t_{l+1}-t_{l}} \label{mupdate+} \\
m_k(W^+)&=&\max\left[m_{k-1}(W),\frac{x_{l+1}-x_{l+1-k}-2\delta}{t_{l+1}-t_{l+1-k}}\right], \nonumber\\
  &&k=2,\ldots R+1. \nonumber
\eea
Repeating the same steps for $M_k(W^+)$, we obtain
\bea
M_1(W^+)&=&\frac{x_{l+1}-x_{l}+2\delta}{t_{l+1}-t_{l}}  \label{Mupdate+} \\
M_k(W^+)&=&\min\left[M_{k-1}(W),\frac{x_{l+1}-x_{l+1-k}+2\delta}{t_{l+1}-t_{l+1-k}}\right], \nonumber\\
  &&k=2,\ldots R+1.\nonumber
\eea
Eqs. (\ref{mupdate+}) and (\ref{Mupdate+}) allow one to efficiently calculate $\bar m(W^+)$ and $\bar M(W^+)$ from $\bar m(W)$ and $\bar M(W)$. Together with (\ref{mW}), (\ref{MW}) it leads to efficient calculation of $m(W^+)$, $M(W^+)$.

\subsubsection{Removing a point from the left: $W\rightarrow W^-$} This case
\be  m_k(W^-)= \max_{\begin{array}{c}
                      j=l-k \\
                      i=j,\ldots l
                    \end{array}
}\left (\frac{x_i-x_j-2\delta}{t_i-t_j}\right ),
\ee
for $k=1,\ldots R-1$
Hence,
\be \label{mupdate-}m_{k}(W^-)=m_{k}(W),\ \ k=1, \ldots R-1. \ee
In the same way,
\be \label{Mupdate-}M_{k}(W^-)=M_{k}(W),\ \ k=1, \ldots R-1. \ee
Eqs. (\ref{mupdate-}) and (\ref{Mupdate-}) allow one to efficiently calculate $\bar m(W^-)$ and $\bar M(W^-)$ from $\bar m(W)$ and $\bar M(W)$. Together with (\ref{mW}), (\ref{MW}) it leads to efficient calculation of $m(W^-)$, $M(W^-)$.

\subsubsection{Calculation of least squares fit}
One can solve the constrained least-squares problem \eqref{LMS} using existing algorithms, e.g., available in MATLAB. However, they may not be suitable for a real-time implementation on a dedicated hardware.
Since \eqref{LMS} is a special case of the generic least squares fit problem, it allows an analytic solution presented below, which can be implemented in a real-time setting.  

Consider window $W=W_l(R)$. Let $T:=(t_{l-R}-t_l, \ldots, t_{l-1}-t_l, 0)^\top$, $X:=(x_{l-R}, \ldots, x_{l-1}, x_{l})^\top$, $\tau =\sum_{i=0}^R t_{l-i}$, $\chi =\sum_{i=0}^R x_{l-i}$ $\One:=(1,\ldots, 1)^\top$  Then the solution of the unconstrained problem \eqref{LMS} is given by 
\begin{equation}
\label{unconstrained_solution}
k_l^*=X^\top\Phi, \ \ \ b_l^*=X^\top\Psi, \quad \hbox{ where}
\end{equation}
\begin{equation}
\label{PhiPsi}
\Phi=\frac{T-\frac{\tau}{R+1}\One}{T^\top(T-\frac{\tau}{R+1}\One)},\ \
\Psi= \frac{\One}{R+1} - \Phi\tau.
\end{equation}
For the constrained least squares problem, if  $|b^*_l-x_l|>\delta$, then the solution to the constrained problem \eqref{LMS} is given by
\begin{align}
\tilde b^*_{l}&=x_l+\mathrm{sign}(b^*_l-x_l)\delta,\\
\tilde k^*_{l}&=\frac{X^\top T-\tilde b^*_{l} ((R+1)t_l-\tau)-t_l\chi}{T^\top T+(R+1)t_l^2-2t_l\tau}.
\end{align}
For applications with a fixed sampling rate, the vectors $\Phi$ and $\Psi$ will only depend on the size of the window $R$, and will not depend on $l$. This allows one to pre-compute them upfront, which makes the calculation more efficient. 

\subsection{Properties of ADF filter}

\subsubsection{Tuning and performance} ADF has two parameters: $R_{max}$ and $\delta$. One can choose  $R_{max}$ sufficiently large and control the filter performance by the main parameter $\delta$. The estimate $\hat x_l$ of the signal provided by ADF satisfies, cf. \eqref{LMS},
\be|x_l-\hat x_l|\le \delta.\ee
Thus, for  $\delta$ smaller than the noise amplitude $d$, the output and its estimated derivative will reflect noise. Very large $\delta$ leads to filtering through (\ref{LMS}) over the largest moving window of the size $R_{max}$. The optimal tuning is with $\delta$ right above the noise amplitude $d$. With such tuning, the filter has a wide pass-band for signals of sufficiently high amplitudes relative to $\delta$ and low sensitivity to measurement noise with amplitudes below $\delta$, that is often required in applications.  
  
\subsubsection{Memory requirements}
Given parameter $R_{max}$, ADF requires two arrays of size $1\times R_{max}$ for storing the buffer with the last $R_{max}$ measurements and samples times, two arrays $\bar m$ and $\bar M$ of the size $1\times R_{max}$, as well as several scalar variables for real-time calculations. 


\section{PID control case study}
\label{sec:3}

A standard PID control, see e.g. \cite{Astrom2005advanced} for basics and advanced discussions, is taken as a case study for evaluating the properties and practical applicability of the above introduced adaptive differentiating filter. Recall that any PID-type feedback regulator assumes a first-order time derivative of the output value, which is (more than often) not available through sensor measurements. For an across domains example, any motion control system with PID feedback control would require measurement of the relative velocity, i.e. $\dot{x}(t)$, for which there is hardly any sensor technology suitable for feedback. Thus, a practical differentiating filter, in terms of unavoidable phase-lag characteristics and a subsequently achievable signal-to-noise ratio (SNR), appears crucial for effectiveness of the PID control in countless applications.

\subsection{PID feedback with symmetrical optimum}
\label{sec:3:1}

PID feedback regulator, parameterized by the overall control gain $K_p$ and two time constants -- of the integrator and differentiator $T_i$ and $T_d$ respectively -- can be written in Laplace $s$-domain in a parallel form as, cf. e.g. \cite{franklin2019},
\begin{equation}\label{eq:3:1}
v(s) = K_p \Bigl(1 + \dfrac{1}{T_i s} + T_d s  \Bigr) e(s).
\end{equation}
The control error is $e = r - x$, where $x$ is the measured system output and $r$ is the control reference value. In the following, it is assumed that the latter is at least once differentiable, while $|\dot{r}| < R$ and the upper bound $0 < R < \infty$ is the matter of an application's specification. We also note that the transfer function $C(s) = v(s)/e(s)$ of the PID control \eqref{eq:3:1} is strictly speaking improper, but the application of an (arbitrary) differentiating filter to determine $e(s) s$ will provide additional low-pass properties. Thus and, moreover, due to a parallel connection of the PID control \eqref{eq:3:1} there are no issues with an improper transfer function $C(s)$.
\begin{figure}[h!]
\centering 
\includegraphics[width=0.99\columnwidth]{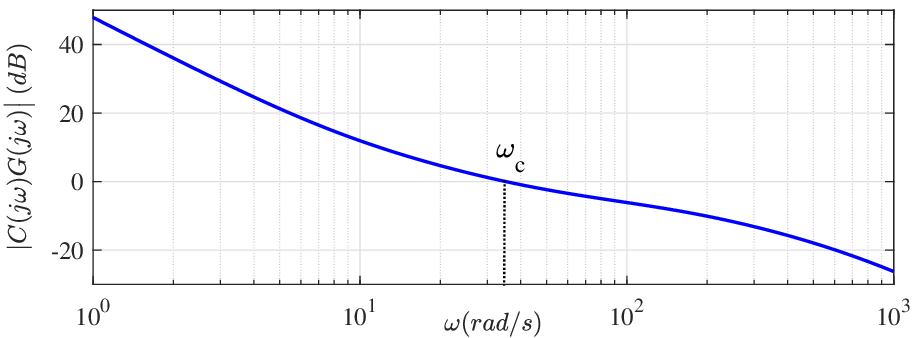}
\includegraphics[width=0.99\columnwidth]{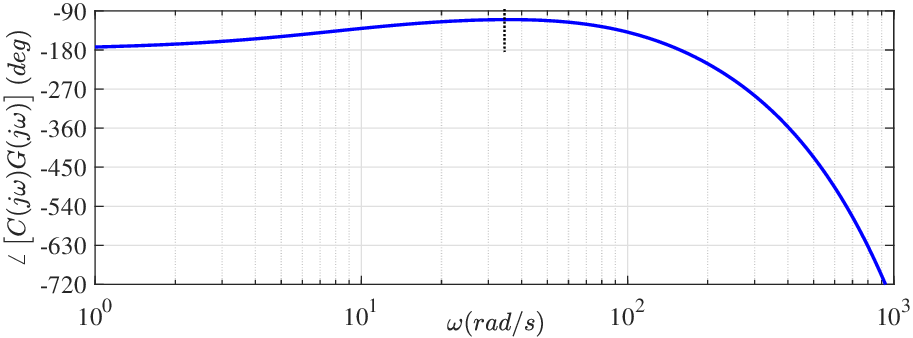}
\caption{Shaped loop transfer function $C(j\omega)G(j\omega)$ with identified system plant $G(j\omega)$ and PID control $C(j\omega)$ designed by symmetrical optimum.}
\label{Fig:LoopOpt} 
\end{figure}

For determining the control parameters $K_p,\, T_i,\, T_d > 0$, we proceed with loop-shaping by following the principles of so-called symmetrical optimum \cite{kessler1958}. The symmetrical optimum, in frequency domain, is achieved by shaping a $-20$ dB/decade slope around the cross-over frequency $\omega_c$, while both corner frequencies left- and right-hand-side from $\omega_c$ are shaped to be approximately equidistant on the logarithmic scale, see Fig. \ref{Fig:LoopOpt}. Worth noting is also that the corresponding system plant transfer function $G(s)$, cf. the present system  in section \ref{sec:3:4}, is equally crucial when shaping the loop by a symmetrical optimum principle. The determined control parameters thus determined are $K_p=420$, $T_i = 0.07$, and $T_d = 0.03$, while it is explicitly emphasized that one and the same designed PID controller \eqref{eq:3:1} is experimentally evaluated below with three different differentiators, each one delivering an $\dot{x}(t)$ approximation. Further we notice that a constant gravity compensating term $\gamma$ is also used, though without affecting the design and analysis, so that the overall control signal is $u(t) = v(t) + \gamma$.

\subsection{Linear filtered differentiation}
\label{sec:3:2}

Direct discrete time differentiation of the measured output value $x(t)$ is considered the first and most intuitive (and also simplest) differentiation approach for comparison. For the given sampling time $T_s$, it leads to a finite differences
\begin{equation}\label{eq:3:2}
\dot{x} \approx \bigl(x(t)-x(t-T_s)\bigr) / T_s.
\end{equation}
Since the measured $x(t)$ contains also some unavoidable (parasitic) noise $w(t)$, and the signal-to-noise ratio (SNR) has usually some finite value $\ll \infty$, the finite difference method \eqref{eq:3:2} becomes more noise-sensitive when $T_s$ decreases. Therefore, an auxiliary low-pass filter (LPF) is required to let the net signal $\dot{x}(t)$ through and, correspondingly, to attenuate the noise related $\dot{w}(t)$ with frequencies higher than some cutoff frequency $\omega_0$, which is the design parameter. In a most simple case, the LPF can be assumed as a 2nd order critically-damped system with unity gain, parameterized by its natural frequency $\omega_0$. Following the lines of designing the control system in continuous time, cf. section \ref{sec:3:1}, while assuming a sufficiently low $T_s$ for an afterwards discrete-time implementation, the linear filtered differentiator (further denoted as LDF) is given by  
\begin{equation}\label{eq:3:3}
\textrm{LDF}(s) = s \, \omega_0^2 \bigl(s^2 + 2 \omega_0 s + \omega_0^2 \bigr)^{-1}.
\end{equation}

\subsection{Robust finite-time differentiator}
\label{sec:3:3}

One of the most widespread robust (and moreover finite-time) differentiators is the sliding-mode based one proposed by Levant \cite{levant1998robust,levant2003} and relying, in its simplest form, on the so-called super-twisting based algorithm (STA). Its remarkable feature is insensitivity to a bounded noise and the finite-time convergence to $n$-th time-derivative of the input signal $x(t)$, provided the Lipschitz constant of its $n$-th time-derivative is available. While intensive research on this class of differentiators was driven in the last two decades, we will solely denote it as robust exact differentiator (shortly  RED), and use its well established parameterization approach developed to a toolbox \cite{reichhartinger2017robust}. The latter has an intuitive and straightforward non-recursive structure and is implementing the homogeneous discrete-time differentiator, as this was proposed in \cite{livne2014proper}, while requiring the single scaling factor $\kappa > 0$ as design parameter. Worth noting is that $\kappa^{n+1}$ corresponds to the Lipschitz constant of the highest derivative ${x}^{(n)}$. The second-order RED is given, cf. \cite{reichhartinger2017robust},
\begin{subequations}\label{eq:RED}
\begin{align}
\dot{z}_0 &= z_1 + 3.1~ \kappa~ \bigl| x-z_0 \bigr|^{\frac{2}{3}} ~\mathrm{sign}(x-z_0),\\[0mm]
\dot{z}_1 &= z_2 + 3.2~ \kappa^2~ \bigl| x-z_0 \bigr|^{\frac{1}{3}} ~\mathrm{sign}(x-z_0),\\[1mm]
\dot{z}_2 &= 1.1 ~ \kappa^3~ \mathrm{sign}(x-z_0),
\end{align}
\end{subequations}
and it provides $z_0(t) = x(t)$, $z_1(t) = \dot{x}(t)$, $z_2(t) = \ddot{x}(t)$ for all $t >
t_c$, where $t_c$ is a finite convergence time. We explicitly notice that the
second-order RED (and not the first-order one) is purposefully used in our study,
in order to obtain a smoother estimate $z_1(t)$ of the output derivative $\dot{x}(t)$.

\subsection{Experimental actuator system}
\label{sec:3:4}

\subsubsection{Hardware setup}
\label{sec:3:4:1}

The used actuator system in a laboratory setting is shown in Fig. \ref{Fig:ExpSetup}. The translational relative motion with a mechanically limited displacement range about $x \in [0, \ldots, 0.018)$ m is actuated by a voice-coil-motor. The real-time control and data acquisition Speedgoat I/O board is running at 2 kHz sampling frequency $f_s = 1/T_s$. The controller output is the terminal voltage $u(t) \in [0, \ldots, 10]$ V of the voice-coil-motor, while its nominal electro-magnetic time constant (given by the manufacturer data-sheet) is $\mu = 0.0012$ sec.   
\begin{figure}[h!]
\centering \includegraphics[width=0.52\columnwidth]{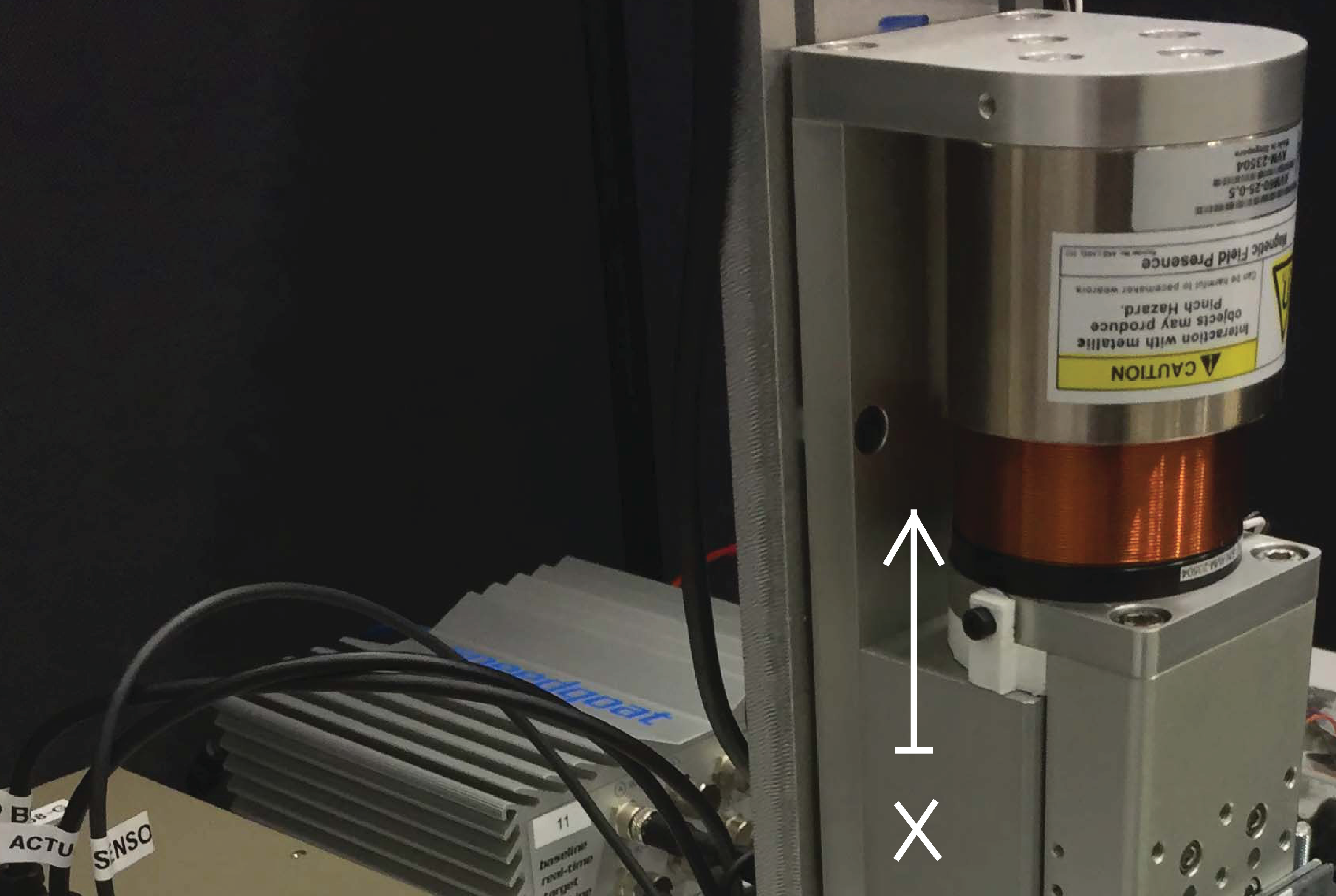}
\caption{Actuator system in laboratory setting with translational motion $x$.}
\label{Fig:ExpSetup} 
\end{figure}
The system output relative displacement $x(t)$ is measured contactless by a down placed inductive distance sensor with a nominal repeat accuracy of $\pm 12$ micrometers. Such sensing location and overall hardware configuration result in a relatively high level of the measurement noise, cf. below in section \ref{sec:3:5}. Further details on the experimental system in use and identification can be found in \cite{ruderman2022motion,ruderman2022disturbance}.

\subsubsection{FRF identified model}
\label{sec:3:4:2}

The input-output system model \eqref{eq:4:1} was accurately identified from the point-wise measured FRF (frequency response function) as shown in Fig. \ref{Fig:FRF}. 
\begin{figure}[h!]
\centering \includegraphics[width=0.99\columnwidth]{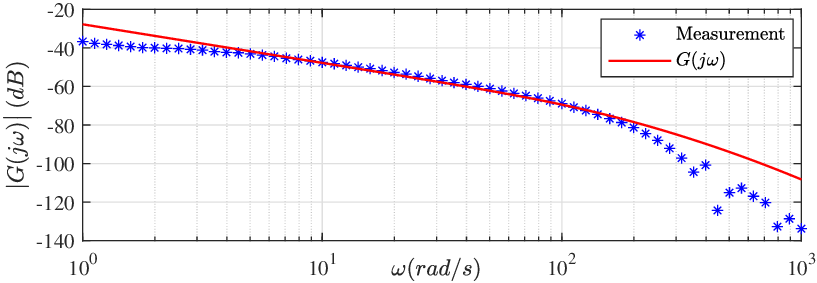}
\centering \includegraphics[width=0.99\columnwidth]{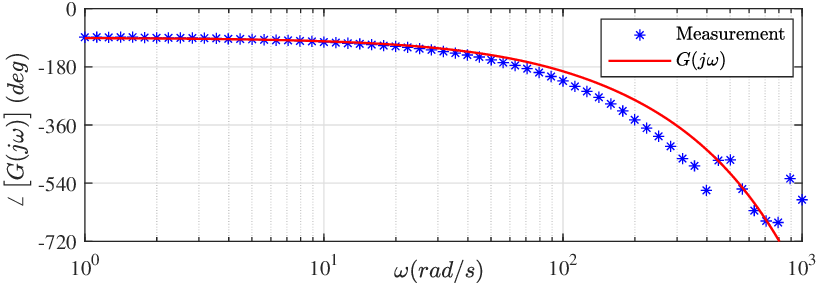}
\caption{Measured and identified FRF of the system.}
\label{Fig:FRF} 
\end{figure}
Note that the total moving mass parameter $m$ is weighed on the electronic scales and thus known, equally as the given electro-magnetic time constant $\mu$, cf. subsection above. Also, the nominal value of the overall actuator gain $K_m$ is known from the technical data-sheets of the voice-coil-motor.
\begin{equation}\label{eq:4:1}
G= \frac{K_m \exp(-\tau s)}{s(m s + \nu)(\mu s + 1)} = \frac{3.28 \exp(-0.011 s)}{0.00064 s^3 + 0.634 s^2 + 80 s}.
\end{equation}
The derived $G(s) = x(s)/u(s)$ model \eqref{eq:4:1} constitutes, this way, the best linear fit in frequency domain. The uncertain and thus experimentally estimated parameters are the linear damping $\nu$ and the time delay factor $\tau$. More details on the accomplished FRF based identification of the system can be found in \cite{ruderman2022disturbance}. Note that the time delay appears largely as not negligible, cf. phase diagram in Fig. \ref{Fig:FRF}, and can be seen as the one lumped delay value accumulated from all converters, hardware channels of data processing, and amplifiers in the loop. Further recall that the transfer function \eqref{eq:4:1} serves as a basis for symmetrical optimum tuning of the PID controller under examination, cf. section \ref{sec:3:1}.

\subsection{Comparison of three differentiators}
\label{sec:3:5}

The ADF is parameterized by taking into account the CPU capacities of the control board with hard real-time constraints and an estimate of the upper bound on the measurement noise $d$, cf. section \ref{sec:2}. An example of the measured $x(t)$ at zero control value is shown in Fig. \ref{Fig:ADFexpest} (a), including its distribution histogram. Allowing for the single impulse-shape outliers, the set parameter is $\delta = 0.0001$ m, while $R_{max} = 140$ for keeping the real-time data processing manageable and, therefore, to avoid a CPU overrun during the execution.    
\begin{figure}[h!]
\centering \includegraphics[width=0.98\columnwidth]{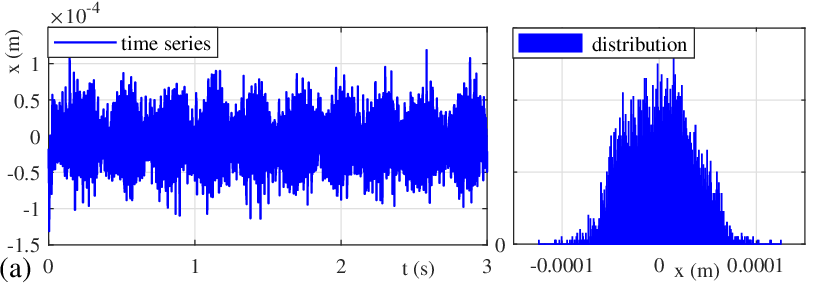}
\centering \includegraphics[width=0.99\columnwidth]{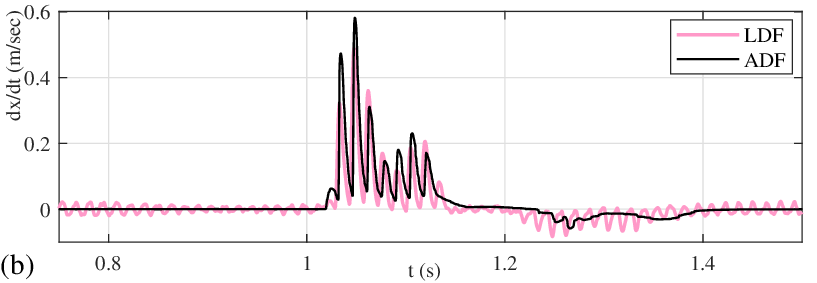}
\caption{Measured noisy output and its distribution at the idle state ($u=0$) in (a), and exemplary $\dot{x}(t)$ estimate by means of LDF and ADF in (b).}
\label{Fig:ADFexpest} 
\end{figure}  
The LDF is parameterized by $\omega_0 = 600$ rad/sec. Note that with respect to the designed PID controller, cf. section \ref{sec:3:1}, such low-pass filter (applied for obtaining $\dot{x}(t)$ estimate only) does not significantly change the shaped loop transfer characteristics. Solely the otherwise sufficient phase margin is reduced by less than $5$ deg, cf. Fig. \ref{Fig:LoopOpt}. An exemplary estimate of the output velocity, obtained by means of the designed LDF and ADF, are shown opposite to each other in Fig. \ref{Fig:ADFexpest} (b) for the measured $x(t)$ during the controlled step response, cf. below. Finally, the RED is parameterized by setting $\kappa = 8$, that is experimentally tuned as provided in \cite{ruderman2022disturbance} on an up-chirp signal excitation response within one decade (e.g. between 8 rad/sec and 80 rad/sec).  

For the sake of a better traceability, the numerically determined FRFs of all three differentiators are shown for two and a half decades in Fig. \ref{Fig:DiffFRF}, opposed to each other and to an exact differentiation, i.e. $|j\omega|$. Note that waves in the shape at lower frequencies are associated with chirp-signal used for a smooth input. 
Therefore, it allows for an only quasi-steady-state frequency consideration, especially in the fist decade. Apart from that, one can recognize that all three differentiators provide an expected 20 dB/dec increase over the total frequency range of interest. To notice is that the RED breaks down around $\omega = 100$ rad/sec, -- a bandwidth for which it was designed by assigning the scaling factor $\kappa$.
\begin{figure}[h!]
\centering \includegraphics[width=0.99\columnwidth]{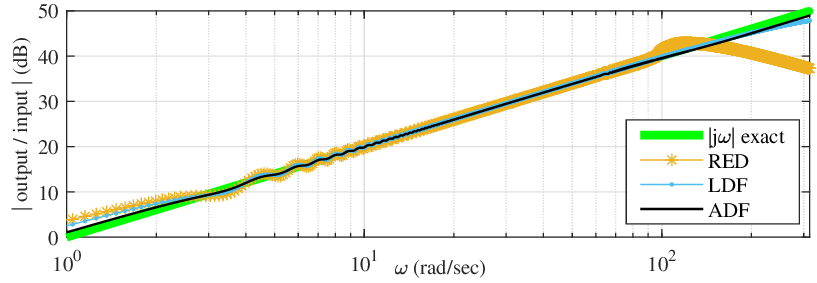}
\caption{Numerically determined FRF of three differentiators.}
\label{Fig:DiffFRF} 
\end{figure}  

The designed PID control (cf. section \ref{sec:3:1}) is experimentally evaluated with three different differentiators, LPF, ADF, and RED, once on a step and once on a slope reference, as shown in Fig. \ref{Fig:ControlResults}.    
\begin{figure}[h!]
\centering 
\includegraphics[width=0.99\columnwidth]{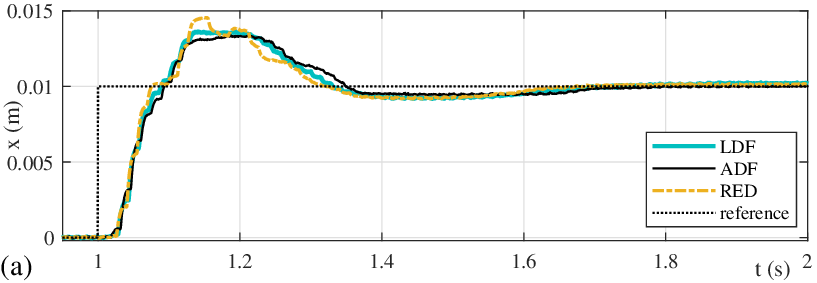}
\includegraphics[width=0.99\columnwidth]{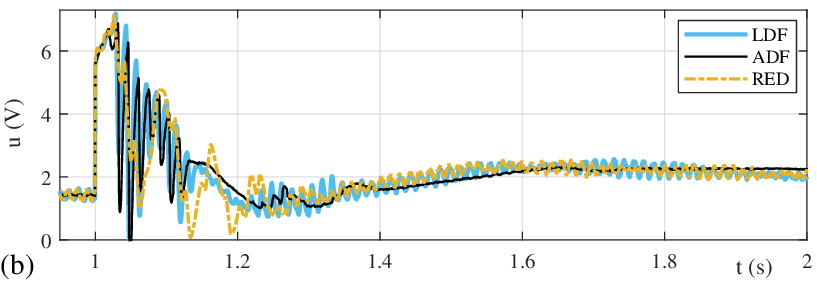}
\includegraphics[width=0.99\columnwidth]{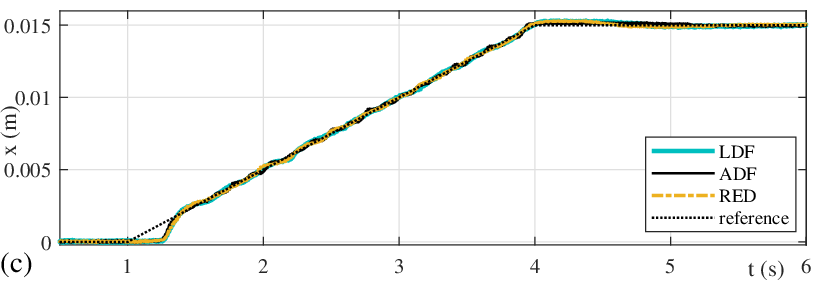}
\includegraphics[width=0.99\columnwidth]{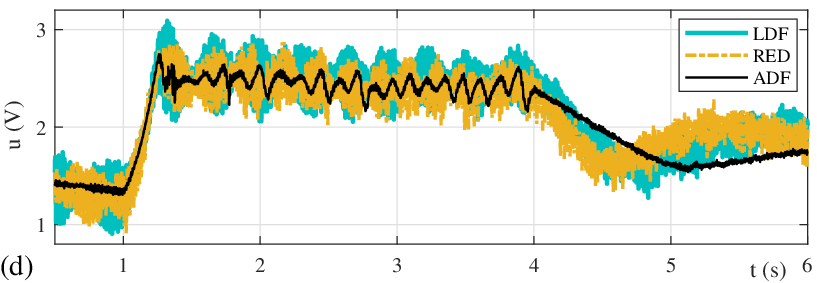}
\caption{Measured PID response (with LDF, ADF, RED) of output and control values: (a), (b) to step reference and (c), (d) to slope reference.}
\label{Fig:ControlResults} 
\end{figure}  
All three differentiation schemes lead to a relatively similar control response, while ADF provides a slightly lower overshoot in case of the step response. Also to notice is a time delay, cf. $\tau$ in section \ref{sec:3:4:2}, and a slow settling of the optimal PID control, cf. Fig. \ref{Fig:ControlResults} (a), due to the nonlinear Coulomb-type friction, see \cite{ruderman2025} for details. An essential advantage of using ADF differentiator becomes evident by inspecting $u(t)$ control values, cf. Fig. \ref{Fig:ControlResults} (b) and (d). While being in principle uniform, the higher oscillatory control patterns (and so energy consumption and system loading) are clearly with LDF and RED due to richer frequency components of $\dot{x}(t)$ estimate, cf. also Fig. \ref{Fig:ADFexpest} (b).

\section{Concluding statements}
\label{sec:4}

A novel causal discrete-time filter based on a constrained least-squares estimator with window adaptation was proposed for differentiating noisy process measurements. It shows low sensitivity to the measurement noise while preserving large amplitude and high-dynamic changes in the process signal. The filter is real-time capable and shows excellent performance in differential feedback control action, evaluated and compared on a real physical motion control application. More detailed analysis of the filter sensitivity and convergence properties, transfer characteristics, and optimal tuning are subject of our upcoming works in progress.

\bibliographystyle{IEEEtran}
\bibliography{references}

\end{document}